\documentclass{article}[12pt]
\usepackage{amsmath,amsthm,amssymb,amsfonts,amscd,eucal,enumerate}

\numberwithin{equation}{section}

\def\bn{{\mathbb N}}

\def\r{\rho}

\def\itm#1{\item{$(#1)$}}

\newtheorem{theorem}{Theorem}[section]

\newtheorem{proposition}[theorem]{Proposition}

\newtheorem{definition}{Definition}[section]
\newtheorem{example}{Example}[section]

\theoremstyle{definition}
\newtheorem{remark}{Remark}[section]

\makeindex

\begin{document}

\title{On a correspondence between regular and non-regular operator monotone functions}

\author{
P. Gibilisco \footnote{Dipartimento SEFEMEQ, Facolt\`a di
Economia, Universit\`a di Roma ``Tor Vergata", Via Columbia 2, 00133
Rome, Italy.  Email:  gibilisco@volterra.uniroma2.it --
URL: http://www.economia.uniroma2.it/sefemeq/professori/gibilisco},
F.Hansen \footnote{Department of Economics, University of Copenhagen, Studiestraede 6, DK-1455 Copenhagen K, Denmark.  Email: Frank.Hansen@econ.ku.dk --  URL: http://www.econ.ku.dk/okofh/},
T. Isola\footnote{Dipartimento di Matematica,
Universit\`a di Roma ``Tor Vergata",
Via della Ricerca Scientifica, 00133 Rome, Italy.
Email: isola@mat.uniroma2.it --
URL: http://www.mat.uniroma2.it/$\sim$isola}
}

\maketitle

\begin{abstract}

We prove that there is a bijection between the families of regular and non-regular operator monotone functions. As an application we give a new proof of the operator monotonicity of a certain class of functions related to Wigner-Yanase-Dyson skew information.

\smallskip

\noindent 2000 {\sl Mathematics Subject Classification.} Primary 62B10, 94A17; Secondary 46L30, 46L60.

\noindent {\sl Key words and phrases.}  Operator monotone functions, matrix means, quantum Fisher information.

\end{abstract}


\section{Introduction}

In  \cite{WignerYanase:1963} Wigner and Yanase  proposed to find a measure of our knowledge of a difficult-to-measure observable
with respect to a conserved quantity. They discussed a number of postulates that such a measure should satisfy and proposed, tentatively,
the so called {\em skew information} defined by
$$
I_{\rho}(A)=-\frac{1}{2}{\rm Tr}([\rho^{\frac{1}{2}},A]^2),
$$
where $\rho$ is a state (density matrix) and $A$ is an observable (self-adjoint matrix), see the discussion in \cite{Hansen:2006}.
The postulates Wigner and Yanase discussed were all considered essential for such a measure of information and included the requirement from thermodynamics
that knowledge decreases under the mixing of states; or put equivalently, that the proposed measure is a convex function in the state $ \rho. $
Wigner and Yanase were aware that other measures of quantum information could satisfy the same postulates, including the measure
$$
I_{\rho}^{\beta}(A)=-\frac{1}{2}{\rm Tr}([\rho^{\beta},A] \cdot [\rho^{1-\beta},A]) \qquad
$$
with parameter $ \beta $ $ (0<\beta<1) $ suggested by Dyson and today known as the Wigner-Yanase-Dyson skew information.
Even these measures of quantum information are only examples of a more general class of information measures, the so called metric adjusted skew
informations \cite{Hansen:2006}, that all enjoy the same general properties as discussed by Wigner and Yanase for the skew information.

The Wigner-Yanase-Dyson (WYD) measures of information are not only used in quantum information theory. A list of applications in other fields
includes: i) strong subadditivity of entropy \cite{LiebRuskai:1973,Lieb:1973}; ii) homogeneity of the state space of factors
of type ${\rm III}_1$  \cite{ConnesStormer:1978};  iii) measures for quantum entanglement \cite{Chen:2005,KlyachkoOztopShumovsky:2006}; iv) uncertainty relations (see \cite{AudenaertCaiHansen:2008,GibiliscoIsola:2008c} and references therein); v) hypothesis testing \cite{CMTMAB:2008}.

This is, in a certain sense, not surprising since the WYD-information is connected to special choices of quantum Fisher information
(see \cite{HasegawaPetz:1997, Hansen:2006}). Similarly, the classical Fisher information was born inside statistics but now plays an important role in a manifold of different mathematical fields, some very far from the original statistical arena (see, for example, \cite{Carlen:1991}).

The crucial ingredient in establishing the connection between the WYD-information and quantum Fisher information is to prove operator monotonicity of the function
\begin{equation} \label{beta}
f_{\beta}(x)= \beta (1-
\beta) \frac{(x-1)^2}{(x^{\beta}-1) (x^{1-\beta}-1)}\qquad \beta
\in (0,1),
\end{equation}
see \cite{HasegawaPetz:1997, Hansen:2006,Szabo:2007} for the existing proofs. We will show that this fact is a simple corollary of a more general result that represents the main goal of the present paper.

To explain the main result of the paper we have to recall that in the last century fundamental bijections have been established between a certain family
of operator monotone functions, the Kubo-Ando operator means and the various types of quantum Fisher information (see \cite{Lowner:1934, KuboAndo:1979/80, Petz:1996}).

 Each group of objects can be subdivided into two components according to what follows. Any quantum Fisher information can be seen as a Riemannian metric on the space of faithful states (density matrices). It is natural to ask in which cases one can (radially) extend this Riemannian metric on the complex projective space given by the pure states. It turns out that this is possible if and only if the associated operator monotone function is {\em regular}, namely if $f(0)>0$
 (see \cite{Hansen:2006,PetzSudar:1996}). In this case the radial limit is just a multiple of the Fubini-Study metric.

Completing a work started in \cite{GibiliscoImparatoIsola:2007} we prove in Section \ref{tilde} that the application $f \to \tilde f,$ where
 $$
\tilde{f}(x)=\frac{1}{2}\left[ (x+1)-(x-1)^2 \frac{f(0)}{f(x)}
\right]\qquad x>0,
$$
is a bijection between the regular and the non-regular operator monotone functions. The operator monotonicity of the functions in (\ref{beta})
 then easily follows from the main result.

\section{Operator monotone functions, matrix means and quantum Fisher informations} \label{prel}

Let $M_n:=M_n(\mathbb{C})$ (resp. $M_{n,sa}:=M_{n,sa}(\mathbb{C})$)
be the set of all $n \times n$ complex matrices (resp.  all $n
\times n$ self-adjoint matrices).  We shall denote general matrices
by $X,Y,\ldots$ while letters $A,B,\ldots$ will be used for
self-adjoint matrices, endowed with the Hilbert-Schmidt scalar
product $\langle A,B \rangle={\rm Tr}(A^*B)$.  The adjoint of a
matrix $X$ is denoted by $X^{\dag}$ while the adjoint of a
superoperator $T:(M_n,\langle \cdot,\cdot \rangle) \to (M_n ,\langle
\cdot,\cdot \rangle)$ is denoted by $T^*$. Let ${\cal D}_n$ be the
set of strictly positive elements of $M_n$ and ${\cal D}_n^1 \subset
{\cal D}_n$ be the set of strictly positive density matrices, namely
$ {\cal D}_n^1=\{\rho \in M_n \vert {\rm Tr} \rho=1, \, \rho>0 \} $.
If not otherwise specified, we shall from now on only consider faithful $(\rho>0)$ states.

A function $f:(0,+\infty)\to
\mathbb{R}$ is said to be {\it operator monotone (increasing)} if, for any
$n\in \bn$ and $A$, $B\in M_n$ such that $0<A\leq B$, the
inequality $f(A)\leq f(B)$ hold.  A positive operator monotone
function $ f $ is said to be {\it symmetric} if $f(x)=xf(x^{-1}),$ and {\it
normalized} if $f(1)=1$.

\begin{definition}
${\cal F}_{op}$ is the class of functions $f: (0,+\infty) \to (0,+\infty)$ such that

\itm{i} $f(1)=1$,

\itm{ii} $tf(t^{-1})=f(t)$,

\itm{iii} $f$ is operator monotone.
\end{definition}

\begin{example}
    Examples of elements of ${\cal F}_{op}$ are given by the following
list
\[\begin{array}{rcllrcl}
f_{\text{RLD}}(x)&=&\displaystyle\frac{2x}{x+1},&&
f_{\text{WY}}(x)&=&\displaystyle\left(\frac{1+\sqrt{x}}{2}\right)^2,\\[12pt]
f_{\text{SLD}}(x)&=&\displaystyle\frac{1+x}{2},&& f_{\beta}(x)&=&\displaystyle \beta (1-\beta) \frac{(x-1)^2}{(x^{\beta}-1) (x^{1-\beta}-1)}\qquad \beta
\in (0,1).
$$\end{array}\]

\end{example}

A very short account of Kubo-Ando's theory of matrix means \cite{KuboAndo:1979/80} may be summarized as follows:

\begin{definition}
 A {\sl mean} for pairs of positive matrices is a function
$m:{\cal D}_n \times {\cal D}_n \to {\cal D}_n$ such that

\begin{enumerate}[(i)]

\item $m(A,A)=A$,

\item $m(A,B)=m(B,A)$,

\item $A <B  \quad \Longrightarrow \quad A<m(A,B)<B$,

\item $A<A', \quad B<B' \quad \Longrightarrow \quad m(A,B)<m(A',B')$,

\item $m$ is continuous,

\item $Cm(A,B)C^* \leq m(CAC^*,CBC^*)$\quad for every $ C \in M_n$.

\end{enumerate}
\end{definition}

Property $(vi)$ is known as the transformer inequality. We denote by
$\displaystyle {\cal M}_{op}$ the set of matrix means. The
fundamental result, due to Kubo and Ando, is the following.

\begin{theorem}
There exists a bijection between ${\cal M}_{op}$ and ${\cal F}_{op}$ given by
the formula
$$
m_f(A,B)= A^{\frac{1}{2}}f(A^{-\frac{1}{2}} B
A^{-\frac{1}{2}})A^{\frac{1}{2}}.
$$
\end{theorem}

If ${\cal N}$ is a differentiable manifold we denote
by $T_{\rho} \cal N$ the tangent space to $\cal N$ at the point
$\rho \in {\cal N}$.  Recall that there exists a natural
identification
 of $T_{\rho}{\cal D}^1_n$ with the space of self-adjoint traceless
 matrices; namely, for any $\rho \in {\cal D}^1_n $
$$
T_{\rho}{\cal D}^1_n =\{A \in M_n|A=A^* \, , \, \hbox{Tr}(A)=0 \}.
$$
A Markov morphism is a completely positive and trace preserving operator $T:
M_n \to M_m$. A {\sl monotone metric} is a family of Riemannian metrics $g=\{g^n\}$
 on $\{{\cal D}^1_n\}$, $n \in \mathbb{N}$, such that
 $$
 g^m_{T(\rho)}(TX,TX) \leq g^n_{\rho}(X,X)
 $$
 holds for every choice of Markov morphism $T:M_n \to M_m$, faithful state $\rho \in
 {\cal D}^1_n,$ and $X \in T_\rho {\cal D}^1_n$.
Usually monotone metrics are normalized in such a way that
$[A,\rho]=0$ implies $g_{\rho} (A,A)={\rm Tr}({\rho}^{-1}A^2)$.
A monotone metric is also called (an example of) {\sl quantum Fisher information} (QFI). This notation
is inspired by Chentsov's uniqueness theorem for commutative monotone metrics \cite{Chentsov:1982}.

Define $L_{\rho}(A)= \rho A$ and $R_{\rho}(A)= A\rho$, and observe
 that they are commuting positive superoperators on $M_{n,sa}$. For any $f\in {\cal F}_{op}$ one may also define the positive (non-linear) superoperator
$m_f(L_{\rho},R_{\rho})$.
The fundamental theorem of monotone metrics may be stated in the following way:

\begin{theorem} (see \cite{Petz:1996}).
    There exists a bijective correspondence between monotone metrics (quantum Fisher informations)
    on ${\cal D}^1_n$ and normalized symmetric operator monotone
    functions $f\in {\cal F}_{op}$.  The correspondence is given by
    the formula
    $$
   \langle A,B \rangle_{\rho,f}={\rm Tr}(A\cdot
    m_f(L_{\rho},R_{\rho})^{-1}(B))
    $$
    for positive matrices $ A $ and $ B. $
\end{theorem}

\section{Regular functions and extendable Fisher informations}

\begin{definition}

For $f \in {\cal F}_{op}$ we define $f(0)=\lim_{x\to 0} f(x).$ We say that a function $f \in {\cal
F}_{op}$ is regular if $f(0) \not= 0$ and non-regular if $f(0)= 0,$ cf.~\cite{PetzSudar:1996,Hansen:2006}.

\end{definition}

\begin{definition}
A quantum Fisher information is extendable if its radial limit exists and is a Riemannian metric on the real projective space generated by the pure states.
\end{definition}

For the definition of radial limit see \cite{PetzSudar:1996} where the following fundamental result is proved:

\begin{theorem}

An operator monotone function $f \in {\cal F}_{op}$ is regular, if and only if
$ \langle \cdot, \cdot \rangle_{\rho,f}$ is extendable.

\end{theorem}

\begin{remark}
The reader should be aware that there is no negative connotation associated with the qualification ``non-regular".
For example, a very important quantum Fisher information for quantum physics (see \cite{FickSauermann:1990}), namely the Kubo-Mori metric (related to the function $f(x)=\frac{x-1}{\log x})$, is non-regular.
\end{remark}

\section{Some preliminary notions}

\begin{definition}
The Morozova-Chentsov function $ c_f $ associated to a function $f \in {\cal F}_{op}$ is given by
$$
c_f(x,y)=\frac{1}{m_f(x,y)}\qquad x,y>0.
$$
If $f$ is regular one can also define the function
$$
d_f(x,y)=\frac{x+y}{f(0)}-(x-y)^2c_f(x,y).
$$
Another useful definition is the following
$$
c_{\lambda}(x,y)=\frac{1+\lambda}{2} \left(\frac{1}{x+\lambda y}+\frac{1}{\lambda x +y} \right)\qquad \lambda \in [0,1].
$$
\end{definition}

In the result that follows we synthesize  Corollaries 2.3, 2.4 and Proposition 3.4 of the paper \cite{Hansen:2006}, see also the beginning of Section 2 in \cite{AudenaertCaiHansen:2008}.

\begin{theorem} \label{resume}
Given $f \in {\cal F}_{op}$ there exist a unique (canonical) probability measure $\mu$ on [0,1] such that
\[
\begin{array}{rl}
\displaystyle\frac{1}{f(t)}&=\displaystyle\int_0^1 c_{\lambda}(t,1) \, d\mu(\lambda)\qquad t>0,\\[3ex]
c_f(x,y)&\displaystyle=\int_0^1 c_{\lambda}(x,y) \, d\mu(\lambda)\qquad x,y>0,\\[3ex]
d_f(x,y)&\displaystyle=\int_0^1 xy \cdot c_{\lambda}(x,y) \, \frac{(1+\lambda)^2}{\lambda}\, d\mu(\lambda)\qquad x,y>0.
\end{array}
\]
Furthermore, $d_f$ is operator concave as a function of two variables.
\end{theorem}

\section{The correspondence $f\to\tilde f$ and its properties} \label{tilde}

We introduce the sets of regular and non-regular functions
$$
{\cal F}_{op}^{\, r}:=\{f\in {\cal F}_{op}| \quad f(0) \not= 0 \},  \quad
{\cal F}_{op}^{\, n}:=\{f\in {\cal F}_{op}| \quad f(0) = 0 \}
$$
and notice that trivially ${\cal F}_{op}={\cal F}_{op}^{\, r}\cup{\cal F}_{op}^{\, n}$\,.

\begin{definition}
For $f \in {\cal F}_{op}^{\, r}$ we set
$$
\tilde{f}(x)=\frac{1}{2}\left[ (x+1)-(x-1)^2 \frac{f(0)}{f(x)}
\right]\qquad x>0.
$$
 We also write ${\cal G}(f)={\tilde f},$ cf. {\rm \cite{Hansen:2006,GibiliscoImparatoIsola:2007,AudenaertCaiHansen:2008}}.
\end{definition}

Notice that one has the identity
$$
\tilde f(x)=\frac{f(0)}{2} \, d(x,1)\qquad x>0.
$$

\begin{theorem}
The correspondence $ f \to \tilde f $ is a bijection between ${\cal F}_{op}^{\, r}$ and ${\cal F}_{op}^{\, n}$.

\end{theorem}

\begin{proof}

Take a function $ f\in{\cal F}_{op}^{\, r} $ and consider $ \tilde f. $ It was noticed in \cite{GibiliscoImparatoIsola:2007} that $\tilde f$ is a non-regular
function in $ {\cal F}_{op}. $ Indeed, it is easy to see that ${\tilde f}(0)=0$, ${\tilde f}(1)=1$ and $x {\tilde f}(x^{-1})=\tilde f(x)$ for $ x>0. $
Furthermore, since $d_f$ is operator concave, so is $\tilde f$. But since a positive operator concave function is operator monotone 
(Theorem 2.5 in \cite{HansenPedersen:1982}) we get the desired conclusion.

It is easy to establish that the correspondence  $ f \to \tilde f $ is injective.

It remains to show that the correspondence $f \to {\tilde f}$ is surjective. Therefore, suppose that $g $ is a non-regular function
in $ {\cal F}_{op}.$ We have to find a regular function $f\in{\cal F}_{op} $ such that ${\tilde f}=g$.
Consider the function
\[
h(t)=\frac{g(t)}{t}=\frac{1}{g^{\sharp}(t)}\qquad t>0,
\]
where $ g\to g^{\sharp} $ is the involution of $ {\cal F}_{op} $ given by
\[
g^{\sharp}(t)=\frac{t}{g(t)}\qquad t>0,
\]
cf. [Definition 2.5] in \cite{AudenaertCaiHansen:2008}. It follows that $ h $ is operator monotone decreasing, $ h(1)=1, $ and
$ h $ satisfies the functional equation
\[
h(t^{-1})=\frac{g(t^{-1})}{t^{-1}}=t\cdot g(t^{-1})=g(t)=t\cdot h(t)\qquad t>0.
\]
Therefore, there exists \cite[Corollary 2.3]{Hansen:2006} a probability measure $ \mu $ on the unit interval
such that
\begin{equation}\label{variant canonical representation for g}
h(t)=\int_0^1 \frac{1+\lambda}{2}\left(\frac{1}{t+\lambda}+\frac{1}{1+t\lambda}\right)d\mu(\lambda)
\qquad t>0.
\end{equation}
Suppose for a moment that $ \mu $ has an atom in zero. Then $ h $ is of the form
\[
h(t)=\mu(0)\frac{t+1}{2t}+\tilde h(t),
\]
where $ \tilde h(t) $ is some non-negative operator monotone function. Consequently,
\[
g(t)=t\cdot h(t)\ge \mu(0)\frac{t+1}{2}\qquad t>0,
\]
contradicting the choice of $ g $ as a non-regular function in $ {\cal F}_{op}. $
We conclude that $ \mu $ has no atom in zero. In particular, if one defines the constant
\[
C=\int_0^1 \frac{2\lambda}{(1+\lambda)^2}\,d\mu(\lambda),
\]
then $C>0$.
As a consequence we may define another probability measure $ \nu $ on the unit interval by setting
\[
d\nu(\lambda)=\frac{1}{C}\cdot \frac{2\lambda}{(1+\lambda)^2}\,d\mu(\lambda).
\]
We now define a function $ f $ in the positive half-axis by setting
\[
\frac{1}{f(t)}:=\int_0^1 \frac{1+\lambda}{2}\left(\frac{1}{t+\lambda}+\frac{1}{1+t\lambda}\right)d\nu(\lambda)
\qquad t>0.
\]
Since the right hand side is operator monotone decreasing, we obtain that $ f $ is operator monotone (increasing). Since also
$ f(1)=1 $ and $ f $ satisfies the functional equation $ f(t)=tf(t^{-1}), $ we realize that
$ f\in{\cal F}_{op}. $ Finally, since the limit
\[
\lim_{t\to 0}\frac{1}{f(t)}=\frac{1}{C}>0,
\]
we conclude that $ f $ is a regular function in $ {\cal F}_{op}\,. $  Note that the measure $d\nu$ coincides with the canonical measure associated to $f$ according to Theorem \ref{resume}.
The function $ \tilde f $  may be written as
\[
\tilde f(t)=\frac{f(0)}{2}\, d_f(t,1)\qquad t>0,
\]
where
\[
d(t,1)=\int_0^1 t\, \frac{1+\lambda}{2}\left(\frac{1}{t+\lambda}+\frac{1}{1+t\lambda}\right)
\frac{(1+\lambda)^2}{\lambda}\, d\nu(\lambda)\qquad t>0.
\]
Inserting $ f(0)=C $ and the measure $ \nu $ we obtain
\[
\tilde f(t)=\frac{f(0)}{2}\, t\int_0^1 \frac{1+\lambda}{2}\left(\frac{1}{t+\lambda}+\frac{1}{1+t\lambda}\right)\,d\mu(\lambda)
=t\cdot h(t)=g(t)\qquad t>0,
\]
where we used that $ \mu $ has no mass in zero.
This ends the proof.
\end{proof}

\section{Some applications}

\subsection{Quantum Fisher information in terms of quantum covariances}

The quantum Fisher information (QFI) is determined, as noted in the standard references on the subject, when we know the metric on the non-commuting part of the tangent spaces. We therefore have to understand what happens for the following kind of scalar products:

\begin{equation}
 \langle i[\rho,A], i[\rho,B]\rangle_{\rho,f}.
\end{equation}
Introducing the {\em quantum $g$-covariance}
$$
{\rm Cov}_{\rho}^g (A,B):={\rm Tr}(m_{g}(L_{\rho},R_{\rho})(A_0)B_0),
$$
where $A_0:= A- {\rm Tr}(\r A)$, one can prove \cite{GibiliscoImparatoIsola:2007} that
\begin{equation}\label{crucial}
\frac{f(0)}{2} \cdot \langle i[\rho,A], i[\rho,B]\rangle_{\rho,f}={\rm Cov}^{f_{\text{SLD}}}_{\rho} (A,B)-{\rm Cov}_{\rho}^{\tilde f} (A,B),
\end{equation}
for regular $f$.

\subsection{Uncertainty principle}

The standard uncertainty principle due to Heisenberg, Schr\"odinger and Robertson (\cite{Heisenberg:1927, Robertson:1929, Schroedinger:1930,
Robertson:1934}) may be formulated as the inequality

\begin{equation} \label{nup}
    {\rm det}\left\{ {\rm Cov}_{\rho}(A_j,A_k) \right\} \geq
    \begin{cases}
	0, & N=2m+1,\\[0.5ex]
	{\rm det} \{-\frac{i}{2}\cdot {\rm Tr}(\rho [A_j,A_k])\}, & N=2m.
    \end{cases}
\end{equation}

This means that for an odd number of observables the above inequality does not say anything more then the classical fact that the correlation matrix of a random vector is positive semidefinite. With the help of formula (\ref{crucial}) one can prove a different inequality
that says that

\begin{equation} \label{dyn}
    {\rm det} \{ {\rm Cov}_{\rho}(A_j,A_k) \} \geq {\rm det} \left\{
    \frac{f(0)}{2} \cdot  \langle i[\rho,A_j], i[\rho,
    A_k]\rangle_{\rho,f}\right\},
\end{equation}
(see \cite{AudenaertCaiHansen:2008, GibiliscoIsola:2008c} and references therein).

\subsection{The inversion formula}

\begin{definition}
For $g \in {\cal F}_{op}^{\, n}$  set 
\begin{equation} \label{gbar}
{\check g}(x)=
	\begin{cases}
		g''(1)\cdot  \frac{(x-1)^2}{2g(x)-(x+1)}, & x\in (0,1)\cup(1,\infty),\\
		1, & x=1.
	\end{cases}
\end{equation}
We also write ${\cal H}(g)={\check g}$.
\end{definition}

\begin{proposition} \label{inversion}

If $g$ is non-regular then ${\check g}$ is regular, namely ${\check g} \in {\cal F}_{op}^{\, r}$.
Moreover if $f \in {\cal F}^{\, r}_{op}$ and $g \in {\cal F}^{\, n}_{op}$ then
$$
{\cal H}({\cal G}(f))=f \qquad \text{and}\qquad  {\cal G}({\cal H}(g))=g.
$$

\end{proposition}

\begin{proof}

Let $g$ be non-regular and $f$ regular such $\tilde f=g$. This means that

$$
g(x)=\frac{1}{2}\left[ (x+1)-(x-1)^2 \frac{f(0)}{f(x)}
\right].
$$
If $x \not=1$ this implies
$$
f(x) =
    -f(0)\cdot  \frac{(x-1)^2}{2g(x)-(x+1)}
$$

Note that the property $x g(x^{-1})= g(x)$ implies that for any $g\in {\cal F}_{op}$ one has $g'(1)=\frac{1}{2}$.

Therefore applying two times the De L'Hopital theorem one has
$$
1=\lim_{x \to 1} f(x)= -f(0)\lim_{x \to 1}\frac{(x-1)^2}{2g(x)-(x+1)}=-f(0) \cdot \frac{1}{g''(1)}.
$$
That is
$$
-f(0)=g''(1).
$$
This ends the proof.
\end{proof}

\subsection{WYD information and a class of operator monotone functions}

The correspondence between the WYD-information
$$
I_{\rho}^\beta(A)=-\frac{1}{2} {\rm Tr}([\rho^{\beta},A][\rho^{1-\beta}, A]),\qquad 0<\beta<1,
$$
and quantum Fisher informations depends, as noted in the introduction, on the operator monotonicity of the functions
$$
f_{\beta}(x)= f_{\text{WYD}(\beta)}(x)= \beta (1-
\beta) \frac{(x-1)^2}{(x^{\beta}-1) (x^{1-\beta}-1)}\qquad 0<\beta<1,
$$
see \cite{HasegawaPetz:1997,Hansen:2006,Szabo:2007} for the existing proofs. We note here that Proposition \ref{inversion} gives a new proof of the above result.

\begin{proposition} The function $ f_{\beta}\in {\cal F}^{\, r}_{op} $ for $ \beta \in (0,1). $
\end{proposition}

\begin{proof}
Note that the function
$$
g_{\beta}(x)=\frac{x^{\beta}+x^{1-\beta}}{2} \qquad 0<\beta<1
$$
is operator monotone and non-regular. Since $ f_{\beta}=\check{g}_{\beta} $ we get the desired conclusion.
\end{proof}


\begin{thebibliography}{10}




\bibitem{AudenaertCaiHansen:2008} Audenaert, K., Cai, L., Hansen, F., Inequality for quantum skew information, arXiv:0803.1056, 2008.

\bibitem{CMTMAB:2008} J. Calsamiglia, R. Munoz-Tapia, L. Masanes, A. Acin and E. Bagan, Quantum Chernoff bound as a measure of distinguishability between density matrices: Application to qubit and gaussian states, {\it Phys. Rev. A}, 77, 032311, (2008).

\bibitem{Carlen:1991}
Carlen E.,
\newblock Superadditivity of Fisher's information and logarithmic Sobolev inequalities.
\newblock {\em J. Funct. Anal.}, 101(1): 194-211, 1991.

\bibitem{Chen:2005} Z. Chen. Wigner-Yanase skew information as tests for quantum entanglement, {\it Phys. Rev. A}, 71, 052302, (2005).


\bibitem{Chentsov:1982}
 {\v{C}}encov, N.~N.,
\newblock {\em Statistical decision rules and optimal inference}.
\newblock American Mathematical Society, Providence, R.I., 1982.

\bibitem{ConnesStormer:1978} A. Connes and E.Stormer. Homogeneity of the state space of factors of type ${\rm III}_1$, {\it J. Funct. Anal.}, 28, 187--196, (1978).

\bibitem{FickSauermann:1990} Fick, E.; Sauermann, G. {\em The quantum statistics of dynamic processes}. Springer Series in Solid-State Sciences, 86. Berlin, 1990.

\bibitem{GibiliscoImparatoIsola:2007}
Gibilisco, P., Imparato, D.  and Isola, T.,
\newblock Uncertainty principle and quantum Fisher information II.
\newblock  {\em J. Math. Phys.}, 48: 072109,
\newblock  2007.

\bibitem{GibiliscoIsola:2008c} Gibilisco, P. and Isola, T., A dynamical uncertainty principle in von Neumann algebras by operator monotone functions, {\em J. Stat. Phys.}, DOI 10.1007/s10955-008-9582-3, 2008.

\bibitem{Hansen:2006}
Hansen, F.,
\newblock Metric adjusted skew information. To appear in {\em Proc. Nat. Acad. Sci. USA},
\newblock  arXiv:math-ph/0607049v3, 2006.

\bibitem{HansenPedersen:1982} Hansen, F. and Pedersen, G.K., Jensen's operator inequality and L{\"o}wner's theorem, {\em Math. Ann.} 258, 229--241, 1982.

\bibitem{HasegawaPetz:1996}
Hasegawa, H. and Petz, D.,
On the Riemannian metric of $\alpha$-entropies of density matrices, {\em Lett. Math Phys.} 38: 221--225, 1996.

\bibitem{HasegawaPetz:1997}
Hasegawa, H. and Petz, D.,
Non-commutative extension of information geometry II, in: O. Hirota {et al.} {\em Quantum Communication, Computing and Measurement}, Plenum 1997.


\bibitem{Heisenberg:1927}
 Heisenberg, W.,
\newblock  {\"U}ber den anschaulichen inhalt der quantentheoretischen kinematik und mechanik.
\newblock {\em Zeitschrift f{\"u}r Physik,} 43:172-198, 1927.

\bibitem{KlyachkoOztopShumovsky:2006} A. Klyachko, B. Oztop and A. S. Shumovsky. Measurable entanglement, {\it Appl. Phys. Lett.}, 88, 124102, (2006).



\bibitem{KuboAndo:1979/80}
Kubo, F. and Ando, T.,
\newblock Means of positive linear operators.
\newblock {\em Math. Ann.,} 246(3): 205--224, 1979/80.

\bibitem{Lieb:1973} E. Lieb, Convex trace functions and the Wigner-Yanase-Dyson conjecture. Advances in Math. 11, 267--288, (1973) .

\bibitem{LiebRuskai:1973} E. Lieb and M.B. Ruskai, A fundamental property of the quantum mechanical entropy, {\it Phys. Rev. Lett.} 30, 434436, (1973).

\bibitem{Lowner:1934}
L\"owner, K. \"Uber monotone Matrixfunktionen. {\em Math. Z.} 38: 177-216, 1934.



\bibitem{Petz:1996}
Petz, D.,
\newblock Monotone metrics on matrix spaces.
\newblock {\em Linear Algebra Appl.}, 244:81--96, 1996.

\bibitem{PetzSudar:1996}
Petz, D.~ and Sud\'ar, C.,
\newblock Geometry of quantum states.
\newblock {\em J. Math Phys.}, 37:2662--2673, 1996.

 \bibitem{Robertson:1929}
Robertson, H. P.
\newblock The uncertainty principle.
\newblock {\em Phys. Rev.} 34: 573-574, 1929.

\bibitem{Robertson:1934}
Robertson, H. P.
\newblock An indeterminacy relation for several observables and its classical interpretation.
\newblock {\em Phys. Rev.} 46: 794--801, 1934.

\bibitem{Schroedinger:1930}
Schr{\"o}dinger, E.,
\newblock About Heisenberg uncertainty relation (original annotation by  Angelow A. and Batoni M. C.).
\newblock {\em Bulgar. J. Phys.} 26 (5--6): 193--203 (2000), 1999. Translation of {\em Proc. Prussian Acad. Sci. Phys. Math. Sect.} 19 (1930), 296--303.

\bibitem{Szabo:2007}
Szab\'o, V.E.S.,
A class of matrix monotone functions, {\em Lin. Alg. Appl.}, 420:79--85, 2007.

\bibitem{WignerYanase:1963}
E. P. Wigner and M. M. Yanase,
\newblock Information contents of distributions.
\newblock {\em Proc. Nat. Acad. Sci. USA} 49: 910--918, (1963).


\end{thebibliography}
\end{document}